\documentclass{article}
\usepackage{ijcai25}

\usepackage{times}
\usepackage{booktabs}

\usepackage[utf8]{inputenc}
\usepackage[ngerman,english]{babel}
\babeltags{german=ngerman}
\usepackage{amssymb}
\usepackage{amsthm}
\usepackage{thmtools}
\usepackage{fancyvrb}
\usepackage{mathtools}
\usepackage{amsmath}
\usepackage{ifthen}
\usepackage{xspace}
\usepackage[hidelinks]{hyperref}
\usepackage{makeidx}
\usepackage{graphicx}
\usepackage{enumitem}
\usepackage{stmaryrd}
\usepackage[linesnumbered,ruled,vlined]{algorithm2e}
\usepackage[small]{caption}
\usepackage[layout=footnote,marginclue,draft]{fixme}
\FXRegisterAuthor{sg}{asg}{SG}	
\FXRegisterAuthor{ls}{als}{LS}	
\FXRegisterAuthor{pw}{apw}{PW}	

\declaretheorem[name=Definition,style=definition,numberwithin=section]{definition}
\declaretheorem[name=Example,style=definition,sibling=definition]{example}

\declaretheorem[name=Remark,style=definition,sibling=definition]{remark}
\declaretheorem[name=Convention,style=definition,sibling=definition]{convention}

\declaretheorem[name=Theorem,sibling=definition]{theorem}

\declaretheorem[sibling=definition]{lemma}

\newenvironment{proofsketch}{%
	\proof}{\endproof}

\newcommand{\ALC}{\mathcal{ALC}}
\newcommand{\EXP}{\textsc{ExpTime}\xspace}
\newcommand{\NEXP}{\textsc{NExpTime}\xspace}
\newcommand{\PSPACE}{\textsc{PSpace}\xspace}

\title{Non-expansive Fuzzy $\ALC$}
\author{Stefan Gebhart \and Lutz Schr\"oder \And Paul Wild\\
	\affiliations
	Friedrich-Alexander-Universit\"at Erlangen-N\"urnberg\\
	\emails
	\{stefan.gebhart, lutz.schroeder, paul.wild\}@fau.de}

\begin{document}
	\maketitle
	
	\pagestyle{empty}
	\begin{abstract}
		Fuzzy description logics serve the representation of vague
		knowledge, typically letting concepts take truth degrees in the unit
		interval. Expressiveness, logical properties, and complexity vary
		strongly with the choice of propositional base. The \L{}ukasiewicz
		propositional base is generally perceived to have preferable logical
		properties but often entails high complexity or even
		undecidability. Contrastingly, the less expressive Zadeh
		propositional base comes with low complexity but entails essentially
		no change in logical behaviour compared to the classical case. To
		strike a balance between these poles, we propose \emph{non-expansive
			fuzzy $\ALC$}, in which the Zadeh base is extended with
		\L{}ukasiewicz connectives where one side is restricted to be a
		rational constant, that is, with constant shift operators. This
		allows, for instance, modelling dampened inheritance of properties
		along roles. We present an unlabelled tableau method for
		non-expansive fuzzy $\ALC$, which allows reasoning over general
		TBoxes in \EXP like in two-valued~$\ALC$.
		
	\end{abstract}
	\section{Introduction}
	Fuzzy description logics (fuzzy DLs; see~\cite{LukasiewiczStraccia08}
	for an overview) model vague knowledge by replacing the classical
	two-valued truth set with more fine-grained alternatives, most
	commonly with the unit interval. (Other choices are possible, such as
	quantales, e.g.~\cite{WildSchroder21}, or MV-algebras,
	e,g.~\cite{FlaminioEA13}). Both concepts and roles may be `fuzzified'
	in this sense, allowing for the appropriate representation of concepts
	such as `tall person' or `fast car', and of roles such as `likes' or
	`supports'. In the design of fuzzy DLs, an important factor is the
	choice of propositional base, that is, of the set of propositional
	connectives and their interpretation over the given truth set. Over
	the unit interval $[0,1]$, standard options include Zadeh,
	\L{}ukasiewicz, Gödel, and product logic~\cite{LukasiewiczStraccia08},
	with the Zadeh and the \L{}ukasiewicz base having received a
	comparatively large share of the overall
	attention~\cite{BobilloStraccia11,BorgwardtPenaloza17,BouEA11,Hajek05,KulackaEA13,StoilosEA07,Straccia05}.
	
	In the Zadeh base, conjunction and disjunction are just interpreted as
	minimum and maximum, respectively. This has intuitive appeal but
	closer analysis has shown that under this interpretation,  the logic in fact remains very close to its two-valued correspondent. For
	instance, the problem of deciding whether an $\ALC$ concept is
	satisfiable with truth degree at least~$p$ in Zadeh semantics is
	equivalent to satisfiability in two-valued semantics if $p\ge 0.5$,
	and largely trivial, in particular decidable in linear time, if
	$p<0.5$~\cite{BonattiTettamanzi03,KellerHeymans09,Straccia01}. The
	\L{}ukasiewicz base, which uses the additive structure of $[0,1]$ to
	interpret disjunction and conjunction, does not suffer from such
	deficiencies, and is generally perceived to have favourable logical
	properties (in fact, it is determined up to isomorphism by a set of
	desirable properties including residuation and
	axiomatizability~\cite{KunduChen98}). On the other hand, the
	\L{}ukasiewicz base comes with an increase in computational
	hardness. Indeed, the best known upper bound for concept
	satisfiability in \L{}ukasiewicz fuzzy $\ALC$ is
	$\NEXP$~\cite{KulackaEA13,Straccia05} (compared to $\PSPACE$ for
	two-valued $\ALC$~\cite{Ladner77}), and reasoning over general TBoxes
	in \L{}ukasiewicz fuzzy $\ALC$ is even undecidable~\cite{BaaderPenaloza11}.
	
	In the present work, we aim to strike a balance between these poles,
	proposing \emph{non-expansive fuzzy $\ALC$} as a logic that offers
	more expressive power than Zadeh fuzzy $\ALC$ but retains the same
	complexity as two-valued $\ALC$. As the propositional base, we use an
	extension of the Zadeh base where we allow rational truth constants
	and \L{}ukasiewicz connectives with one argument restricted to be a
	truth constant; phrased more simply, the latter effectively just means
	that we include constant shift operators $(-)\oplus c$ where~$\oplus$
	denotes truncated addition and~$c$ is a rational constant (one can
	then also express truncated subtraction $(-)\ominus c$). For instance,
	the TBox axiom
	\begin{equation*}
		\mathsf{Rich}\sqsubseteq \forall\,\mathsf{hasChild}.\,(\mathsf{Rich}\ominus 0.1)
	\end{equation*}
	asserts (debatably, of course) that children of rich people tend to be
	even richer than their parents. This propositional base has been
	widely used in modal logics that characterize behavioural distances, for
	instance, on probabilistic~\cite{BreugelWorrell05} or
	fuzzy~\cite{WildSchroderEtAl18} systems, and in particular does ensure
	non-expansiveness of the logic w.r.t.\ behavioural distance; hence our
	choice of nomenclature.
	
	Our main technical result on non-expansive fuzzy $\ALC$ is
	decidability of the main reasoning problems in the same complexity as
	for two-valued $\ALC$; most notably, threshold satisfiability over
	general TBoxes is (only) $\EXP$-complete, in sharp contrast with the
	undecidability encountered for full \L{}ukasiewicz fuzzy $\ALC$, and
	in spite of the fact that the semantics of general concept inclusions
	is pointwise inequality and thus corresponds to validity of
	\L{}ukasiewicz implication. We base this result on an unlabelled
	tableau calculus. We construct tableaux using an algorithm that
	follows the global caching principle~\cite{GoreNguyen07} and thus can
	terminate before the tableau has been fully expanded, offering a
	perspective for practical scalability.
	
	Proofs are sometimes omitted or only sketched in the main body; full proofs can be found in the appendix. 
	
	For the published version of this paper, see \cite{ijcai2025p502}.
	
	\paragraph{Further related work} 
	
	The idea of using explicit (rational) truth constants comes from (rational) Pavelka logic~\cite{Hajek95,Pavelka79}
	For \emph{finite-valued} \L{}ukasiewicz fuzzy $\ALC$, the threshold
	satisfiability problem is
	$\PSPACE$-complete~\cite{BouEA11}.
	Reasoning in fuzzy $\ALC$ with product semantics has also been shown to be computationally hard: Validity over the empty TBox is decidable but only a lower bound is known \cite{CeramiEsteva22} and undecidable under general concept inclusions \cite{BaaderPenaloza11}. 
	For a more general classification of the decidability of description
	logics under TBoxes depending on the propositional base, see, for
	example, \cite{Baader2017}.  Under the G\"odel propositional base,
	threshold satisfiability (without a TBox) remains in $\PSPACE$
	\cite{CaicedoEA17}, and decidability is retained even for expressive
	fuzzy description logics \cite{BorgwardtEA16}.  Reasoning in Zadeh
	fuzzy $\ALC$ becomes more involved in presence of ABoxes with
	explicit thresholds. An existing tableau algorithm for this
	case~\cite{StoilosEA06} (which does cover general concept inclusions
	in the sense indicated above) is quite different from ours; in
	particular, it updates labels of tableau nodes after their creation,
	and relies on almost all relevant threshold values being explicitly
	mentioned in the ABox. Reasoning methods for Zadeh fuzzy DLs have
	been extended to highly expressive
	DLs~\cite{StoilosEA07,StoilosEA14}.  In the absence of TBoxes, there
	is a tableau algorithm for fuzzy $\ALC$ for any continuous t-norm
	\cite{Baader2015}. The tableaux algorithm for \L{}ukasiewicz fuzzy
	$\ALC$ necessarily applies only the case without TBoxes, and works
	in a different way from ours, in particular is labelled. A
	preliminary variant of the tableau calculus for non-expansive fuzzy
	$\ALC$ without TBoxes has featured in an undergraduate thesis
	supervised by the second author~\cite{Hermes23}.
	\section{Non-Expansive Fuzzy $\ALC$}
	
	We proceed to introduce the fuzzy DL \emph{non-expansive fuzzy
		$\ALC$}. As indicated earlier, it follows the Zadeh interpretation
	of conjunction~$\sqcap$ and disjunction~$\sqcup$ as minimum and
	maximum, but includes constant shifts $(-)\ominus c$ and $(-)\oplus c$
	for rational constants as additional propositional operators. Formal
	definitions are as follows.
	\begin{convention}
		Throughout, let ${\triangleleft} \in \{ <, \leq \}$,
		${\triangleright} \in \{ >, \geq \}$ and
		${\bowtie} \in \{<, \leq, >, \geq\}$. Whenever we talk about
		constants we refer to rational numbers, usually in the unit interval. We encode a constant by taking the binary representation of the numerator and denominator of its representation as an irreducible fraction.
	\end{convention}
	\begin{definition}
		\begin{enumerate}[wide]
			\item A \emph{signature} of a description language consists of a set $\mathsf{N_C}$ of \emph{atomic concepts} and a set $\mathsf{N_R} := \{R_i \mid i \in I\}$ of \emph{role names} for some index set $I$.
			\item Let $(\mathsf{N_C}, \mathsf{N_R})$ be a signature. Then \emph{concepts} $C,D,\dots$ of
			\emph{non-expansive fuzzy $\ALC$} are generated by the grammar
			\begin{equation*}
				C, D ::= p \mid c \mid \lnot C \mid C \ominus c \mid C \sqcap D \mid \exists R . C
			\end{equation*}
			where $c \in [0,1]$ is a constant, $p \in \mathsf{N_C}$ is an atomic concept and $R \in \mathsf{N_R}$ is a role name.
			\item A (fuzzy) \emph{interpretation} $\mathcal{I}$ consists of a set $\Delta^\mathcal{I} \neq \emptyset$ of \emph{individuals} , a map $p^\mathcal{I} : \Delta^\mathcal{I} \rightarrow [0,1]$ for every atomic concept $p \in \mathsf{N_C}$, and a map $R^\mathcal{I} : \Delta^\mathcal{I} \times \Delta^\mathcal{I} \rightarrow [0,1]$ for every role name $R \in \mathsf{N_R}$.
			\item Let $\mathcal{I}$ be an interpretation and $x \in \Delta^\mathcal{I}$ an individual. Then we define the valuation of concepts in $x$ with respect to the interpretation $\mathcal{I}$ recursively by
			\begin{equation*}
				c^\mathcal{I}(x)=c \qquad
				(\lnot D)^\mathcal{I}(x)=1-D^\mathcal{I}(x)
			\end{equation*}
			\begin{equation*}
				(C \ominus c)^\mathcal{I}(x)=\max(C^\mathcal{I}(x)-c, 0)
			\end{equation*}
			\begin{equation*}
				(C \sqcap D)^\mathcal{I}(x) = \min(C^\mathcal{I}(x), D^\mathcal{I}(x))
			\end{equation*}
			\begin{equation*}
				(\exists R . C)^\mathcal{I}(x) = \sup_{y \in \Delta^\mathcal{I}}\{\min(R^\mathcal{I}(x,y), C^\mathcal{I}(y))\} .
			\end{equation*}
			\item A \emph{concept assertion} is an inequality of the form $C \mathrel{\bowtie} c$. A \emph{(tableau) sequent} is a finite set of concept assertions.
			\item We define the \emph{(syntactic) size} $\lvert C \rvert$ of a concept $C$ inductively by
			\begin{equation*}
				\lvert p \rvert = 1 \quad \lvert \textstyle\frac{a}{b} \rvert = \log_2{a}+\log_2{b}
			\end{equation*}
			\begin{equation*}
				\lvert \lnot C \rvert = \lvert C \rvert + 1 \quad \lvert C \ominus c \rvert = \lvert C \rvert + \lvert c \rvert + 1
			\end{equation*}
			\begin{equation*}
				\lvert C \sqcap D \rvert = \lvert C \rvert + \lvert D \rvert + 1 \quad \lvert \exists R . C \rvert = \lvert C \rvert + 1.
			\end{equation*}
			\item The \emph{(syntactic) size} of a concept assertion $C \mathrel{\bowtie} c$ is $|C|+|c|$, and that of a sequent is the sum of the syntactic sizes of its elements.
		\end{enumerate}
	\end{definition}
	\noindent As usual, we define $C\sqcup D=\neg(\neg C\sqcap\neg D)$ and
	$\forall R . C = \lnot \exists R . \lnot C$. We also define $C \oplus c = \lnot ((\lnot C) \ominus c)$. Furthermore we call $\exists R . C \mathrel{\triangleright} c$ an existential restriction and $\exists R . C \mathrel{\triangleleft} c$ a universal restriction.
	\begin{remark}\label{rem:concepts}
		The concept language of non-expansive fuzzy $\ALC$ agrees
		essentially with a fuzzy modal logic featuring in a quantitative
		modal characterization theorem~\cite{WildSchroderEtAl18}. We could also
		equivalently define non-expansive fuzzy $\ALC$ as \L{}ukasiewicz
		$\ALC$ where we we keep the weak connectives, i.e. the Zadeh
		connectives, and require that at least one of the arguments of the
		\L{}ukasiewicz connectives must be a constant. Recall here that
		in \L{}ukasiewicz semantics, we have, e.g., strong disjunction $\oplus$ interpreted as addition, so our concepts $C\oplus c$ are effectively strong disjunctions with constants.
	\end{remark}
	\begin{remark}
		For most reasoning problems in fuzzy $\ALC$ with a plain Zadeh base, it has been shown  not only that they have the same complexity as the classical counterpart $\ALC$  but also that logical consequence  remains mostly the same \cite{BonattiTettamanzi03,KellerHeymans09,StoilosEA07,Straccia01}. It has also been noted that non-implication of falsity is not a useful notion of knowledge base consistency in this setting, and is in deterministic linear time while threshold satisfiability is $\EXP$-complete~\cite{BonattiTettamanzi03}.
		Non-expansive fuzzy $\ALC$ employs a more expressive concept language than classical $\ALC$ and as such is not formally subject to such phenomena; we aim to illustrate in Example~\ref{expl:tboxes} that there is also a practical gain in expressiveness over Zadeh fuzzy $\ALC$, in particular in connection with general TBoxes.
	\end{remark}
	\begin{remark}
		In analogy to two-valued notions of bisimilarity, one can give
		a natural fixpoint definition of behavioural distance between
		individuals in interpretations~$\mathcal I$. Under this
		distance, the maps
		$C^{\mathcal I}\colon\Delta^{\mathcal I}\to [0,1]$ become
		non-expansive (and in fact even characterize behavioural
		distance), a property that fails in \L{}ukasiewicz
		semantics~\cite{WildSchroderEtAl18}. This motivates the
		nomenclature `non-expansive fuzzy $\ALC$'.
	\end{remark}
	\begin{definition}
		\begin{enumerate}[wide]
			\item A \emph{TBox} $\mathcal{T}$ is a set of \emph{general concept inclusions} or \emph{GCIs} for short, which are of the form $C \sqsubseteq D$. We say that an interpretation $\mathcal{I}$ satisfies the TBox $\mathcal{T}$ if for every GCI $C \sqsubseteq D$ we have:
			\begin{equation*}
				\forall x \in \Delta^\mathcal{I} : C^\mathcal{I}(x) \leq D^\mathcal{I}(x)
			\end{equation*}
			\item We say that a concept assertion $C \mathrel{\bowtie} c$ is $\mathcal{T}$-\emph{satisfiable} if there exists an interpretation that satisfies $\mathcal{T}$ and that has an individual $x \in \Delta^\mathcal{I}$ such that $C^\mathcal{I}(x) \mathrel{\bowtie} c$. We then also say that $C \mathrel{\bowtie} c$ is \emph{satisfied} in $x$.
			\item We say that a sequent $\Gamma$ is $\mathcal{T}$-\emph{satisfiable} if there exists an interpretation that satisfies $\mathcal{T}$ and that has an individual $x$ such that every concept assertion of $\Gamma$ is satisfied in $x$. We then also say that $\Gamma$ is \emph{satisfied} in $x$.
			\item We say a concept assertion $C \mathrel{\bowtie} c$ is $\mathcal{T}$-\emph{valid} if for every interpretation $\mathcal{I}$ that satisfies $\mathcal{T}$ and all individuals $x \in \Delta^\mathcal{I}$ we have $C^\mathcal{I}(x) \mathrel{\bowtie} c$.
			\item We say a sequent $\Gamma$ is $\mathcal{T}$-\emph{valid} if every concept assertion of $\Gamma$ is $\mathcal{T}$-\emph{valid}.
			\item We define the \emph{(syntactic) size} of a GCI by the sum of the syntactic sizes of its concepts and the syntactic size of a TBox as the sum of the syntactic sizes of its elements.
		\end{enumerate}
	\end{definition}
	\begin{remark}
		There is also a notion of fuzzy GCIs of the form $C \sqsubseteq D \geq p$~\cite{BorgwardtEA15}. In \L{}ukasiewicz semantics, this GCI is satisfied by an interpretation $\mathcal{I}$ iff for all $x \in \Delta^\mathcal{I}$ we have $(C \Longrightarrow D)^\mathcal{I}(x) = \min(1, 1-C^\mathcal{I}(x)+D^\mathcal{I}(x)) \geq p$. However this is clearly equivalent to $\min(1, D^\mathcal{I}(x)+(1-p)) \geq C^\mathcal{I}(x)$, which would be the GCI $C \sqsubseteq (D \oplus (1-p))$ in non-expansive fuzzy $\ALC$. Thus we can restrict ourselves to just handling regular GCIs.
	\end{remark}
	In order to ease working with varying comparison operators, we define the following operators to turn inequalities around and to turn strict inequalities into weak inequalities and vice versa:
	\begin{definition}
		Let ${\bowtie} \in \{<, \leq, >, \geq\}$. We define:
		\begin{equation*}
			\mathrel{\bowtie^\circ} :=
			\begin{cases}
				<, & \text{if } {{\bowtie}}={>}\\
				\leq, &\text{if }{{\bowtie}}={\geq}\\
				>, &\text{if }{{\bowtie}}={<}\\
				\geq, &\text{if }{{\bowtie}}={\leq}
			\end{cases}
			\qquad
			\mathrel{\overline{\bowtie}} :=
			\begin{cases}
				\geq, &\text{if }{{\bowtie}}={>}\\
				>, &\text{if }{{\bowtie}}={\geq}\\
				\leq, &\text{if }{{\bowtie}}={<}\\
				<, &\text{if }{{\bowtie}}={\leq}
			\end{cases}
		\end{equation*}
	\end{definition}
	\noindent Clearly, $C \mathrel{\bowtie} c$ is $\mathcal{T}$-valid iff $C \mathrel{\overline{\bowtie}^\circ} c$ is not $\mathcal{T}$-satisfiable. So in order to prove the $\mathcal{T}$-validity of $C \geq c$ or $C \leq c$, we have to check $C < c$ or $C > c$ for $\mathcal{T}$-satisfiability respectively.
	\begin{remark}
		As noted already for the case of Zadeh $\ALC$~\cite{StoilosEA06},
		TBoxes cause substantial additional difficulties in reasoning
		algorithms. One reason for this additional difficulty is that in
		Zadeh-type logics, TBoxes cannot be internalized as valid
		implications: The concept $\neg C\sqcup D$ is satisfied by all
		individuals in an interpretation iff whenever $C$ has value $>0$,
		then $D$ has value~$1$, which is not equivalent to satisfaction of
		the GCI $C\sqsubseteq D$. Contrastingly, in \L{}ukasiwicz semantics,
		the strong disjunction (cf.\ Remark~\ref{rem:concepts})
		$\neg C\oplus D$ holds in every individual of an interpretation iff
		the interpretation satisfies the GCI $C\sqsubseteq D$; that is, we
		can regard TBoxes as demanding satisfaction of top-level
		\L{}ukasiewicz implications. 
	\end{remark}
	\begin{example}\label{expl:tboxes}
		\begin{enumerate}[wide]
			\item To better illustrate and understand the semantics of non-expansive fuzzy $\ALC$, we begin with the following example:
			\begin{gather*}
				A \sqsubseteq \forall R.\, (A \ominus 0.2) \\
				A \ominus 0.2 \sqsubseteq B \ominus 0.3
				\quad\quad
				B \sqsubseteq (\forall R.\, B) \ominus 0.2
			\end{gather*}
			The $\ominus$ on the right hand side of the GCI $A \ominus 0.2 \sqsubseteq B \ominus 0.3$ makes it so~$B$ has to be bigger than the left hand side by at least $0.3$ and on the other hand the $\ominus$ on the left hand side decreases the value of~$A$ which means that the right hand side has to be at most $0.2$ smaller than~$A$. Combined, this tells us that the value of $B$ is bigger than that of~$A$ by at least $0.1$. So we could reformulate this as just $A \sqsubseteq B \ominus 0.1$.
			The GCI $A \sqsubseteq \forall R.\,(A \ominus 0.2)$ tells us that each $R$-successor either has a successorship degree smaller than or equal to~$1$ minus the value of~$A$ or the value of~$A$ in this successor is bigger than the value of~$A$ at this current individual by at least $0.2$.
			The GCI $B \sqsubseteq (\forall R.\, B) \ominus 0.2$ on the other hand tells us that each $R$-successor either has a successorship degree smaller than or equal to $0.8$ minus the value of $B$ or the value of $B$ in this successor is bigger than the value of $B$ of the current individual by at least $0.2$.
			An example of an inference would then be that $(\lnot(A \ominus 0.5)) \sqcup ((\forall R. B) \ominus 0.2) \geq 0.8$ is $\mathcal{T}$-valid; this makes use of the GCIs $A \ominus 0.2 \sqsubseteq B \ominus 0.3$ (or rewritten as $A \oplus 0.1 \sqsubseteq B$) and $B \sqsubseteq (\forall R. B) \ominus 0.2$ and states that either $A$ is smaller than $0.7$ or $B$ is equal to $1$ in all $R$-successors with non-zero successorship degree.
			\item We model social influences on opinions and beliefs (as a disclaimer, we note that neither this example nor the next one are meant as realistic formalizations of the respective domains): In this model, individuals are people, the single role is the interaction $\mathsf{IsFriendsWith}$ (abbreviated as $\mathsf{IFW}$) and as atomic concepts we take opinions people can hold to some degree. Our TBox could then look like this:
			\begin{gather*}
				\forall \mathsf{IFW} . \mathsf{FootballFan} \sqsubseteq \mathsf{FootballFan}\\
				\forall \mathsf{IFW} . (\lnot\mathsf{FootballFan} \ominus 0.4) \sqsubseteq \lnot \mathsf{FootballFan} \oplus 0.2\\
				\exists \mathsf{IFW} . (\mathsf{SportsFan} \ominus 0.3) \sqsubseteq \mathsf{SportsFan} \oplus 0.2\\
				\mathsf{FootballFan} \sqsubseteq \mathsf{SportsFan}\\
				\mathsf{FootballFan} \ominus 0.3 \sqsubseteq \forall \mathsf{IFW} . (\mathsf{FootballFan} \oplus 0.2)
			\end{gather*}
			We can then reason about how much people like sports or football based on who they interact with and to what degree. An inference we could then make (using the first and fourth GCIs) is that $(\exists \mathsf{IFW} . (\lnot \mathsf{FootballFan}) \oplus 0.4) \sqcup \mathsf{SportsFan} \geq 0.7$ is $\mathcal{T}$-valid, which means that either you have a close friend that is not really interested in football or you are a very big sports fan.
			\item We model the influence of scientists: As individuals we have scientists, as roles we have citation and collaboration relations $\mathsf{CitedBy}$ and $\mathsf{CollaboratedWith}$ (or $\mathsf{CW}$ for short) $\mathsf{Influence}$ of a scientist. Here the fuzziness of the roles represents the frequency of collaborations or citations. Our TBox could then look like this:
			\begin{gather*}
				\exists \mathsf{CitedBy} . \mathsf{Influence} \sqsubseteq \mathsf{Influence} \oplus 0.2\\
				\exists \mathsf{CW} . (\mathsf{Influence} \ominus 0.4) \sqsubseteq \mathsf{Influence}\\
				\forall \mathsf{CW} . \mathsf{Influence} \sqsubseteq \mathsf{Influence}\\
				\mathsf{Influence} \sqsubseteq \exists \mathsf{CitedBy} . (\mathsf{Influence} \oplus 0.3)
			\end{gather*}
			We can then reason about how much citations and collaborations affect influence. 
			For instance, from the above we may infer (using the second and fourth GCIs) that $(\forall \mathsf{CW} . (\lnot \mathsf{Influence} \oplus 0.4) \oplus 0.6) \sqcup (\exists \mathsf{CitedBy} . (\mathsf{Influence} \oplus 0.3)) \geq 0.8$ is $\mathcal{T}$-valid, which means the influence of collaborators of a scientist has an impact on the citations of the scientist.
		\end{enumerate}
	\end{example}
	\section{Tableaux Calculus for TBoxes}
	Having introduced our language and the concept of $\mathcal{T}$-satisfiability, we now construct an unlabelled tableau calculus and prove $\EXP$-completeness of determining if a sequent is $\mathcal{T}$-satisfiable.
	\begin{table*}
		\centering
		\begin{tabular}{c}
			\toprule
			Tableau Rules \\
			\midrule
			\vspace{1em}
			$(\text{Ax } 1) \  \frac{S,p \mathrel{\triangleright} c}{\bot} \quad (\text{if } c \mathrel{\overline{\triangleright}} 1, p \in \mathsf{N_C}) 
			\qquad (\text{Ax } 0) \  \frac{S,p \mathrel{\triangleleft} c}{\bot} \quad(\text{if } c \mathrel{\overline{\triangleleft}} 0, p \in \mathsf{N_C}) 
			\qquad (\text{Ax } c) \   \frac{S, c \mathrel{\triangleleft} d}{\bot} \quad (\text{if } c \mathrel{\overline{\triangleleft^\circ}} d) $\\
			\vspace{1em}
			$(\text{Ax } p) \   \frac{S,p \mathrel{\triangleleft} c, p \mathrel{\triangleright} d}{\bot} \quad (\text{if } c \mathrel{\blacktriangleleft_{({\triangleleft}, {\triangleright})}} d, p \in \mathsf{N_C})
			\qquad (\sqcap \mathrel{\triangleright}) \   \frac{S, C \sqcap D \mathrel{\triangleright} c}{S, C \mathrel{\triangleright} c, D \mathrel{\triangleright} c} 
			\qquad (\sqcap \mathrel{\triangleleft})  \  \frac{S,C \sqcap D \mathrel{\triangleleft} c}{S, C \mathrel{\triangleleft} c \quad S, D \mathrel{\triangleleft} c}$\\
			\vspace{1em}
			$(\lnot \mathrel{\bowtie}) \   \frac{S,\lnot C \mathrel{\bowtie} c}{S, C \mathrel{\bowtie^\circ} 1-c} \qquad (\ominus \mathrel{\triangleleft}) \   \frac{S,C \ominus c \mathrel{\triangleleft} d}{S,C \mathrel{\triangleleft} d + c, d \mathrel{\triangleleft^\circ} 0}
			\qquad (\ominus \mathrel{\triangleright}) \   \frac{S,C \ominus c \mathrel{\triangleright} d}{S,C \mathrel{\triangleright} d + c} \quad (\text{if } d \mathrel{\overline{\triangleright}} 0)$\\
			$(\exists R)\  \frac{S, \{ \exists R . D_j \mathrel{\triangleleft_j} d_j \mid 1 \leq j \leq n\}, \exists R . C \mathrel{\triangleright} c}{\{D_j \mathrel{\triangleleft_j} d_j \mid d_j \mathrel{\blacktriangleleft_{({\triangleleft_j}, {\triangleright})}} c, j \in \{1, \ldots, n\}\}, C \mathrel{\triangleright} c, T \geq 1} \quad (\text{if } c \mathrel{\overline{\triangleright}} 0 \text{ and }S \text{ does not contain any } \exists R . D \mathrel{\triangleleft} d)$\\
			\bottomrule
		\end{tabular}
		\caption{Tableau Calculus for $T \geq 1$ and $\Gamma$}
		\label{tab:tableau}
	\end{table*}
	\noindent We start by finding a more easily computable notion of $\mathcal{T}$-satisfiability: Let $A \sqsubseteq B$ be a GCI from a TBox $\mathcal{T}$ and $\Gamma$ be a sequent. Then for any interpretation~$\mathcal{I}$, we have that $A \sqsubseteq B$ is satisfied in~$\mathcal{I}$ iff for all individuals~$x$, there exists a constant~$c$ such that $A^\mathcal{I}(x) \leq c$ and $B^\mathcal{I}(x) \geq c$. The latter formulas can be rewritten as $((\lnot A \oplus c) \sqcap (B \oplus (1-c)))^\mathcal{I}(x) \geq 1$. Checking this formula for all possible $c$ would not yield a terminating algorithm, however as it turns out it suffices to check this for finitely many constants:
	\begin{definition}\label{def:associatedconceptassertion}
		Let $\mathcal{T} = \{C_i \sqsubseteq D_i \mid i=1, \ldots, n\}$ be a TBox
		and $\Gamma$ be a sequent. Let $Z$ be the intersection of the unit
		interval and the additive subgroup of the rationals generated by $1$ and the
		constants appearing in $\Gamma$ and $\mathcal{T}$. Put
		$\epsilon = \frac{1}{2}\min (Z \setminus \{0\})$ and
		$Z' = Z \cup \{z+\epsilon \mid z \in Z \setminus\{1\}\}$. The \emph{concept
			assertion associated to $\mathcal{T}$ and $\Gamma$} is $T \geq 1$
		where
		$T = (\bigsqcap_{i=1}^n \bigsqcup_{z \in Z'}(\lnot C_i \oplus z) \sqcap
		(D_i \oplus (1-z)))$.
	\end{definition}
	\begin{lemma}\label{lemma:interpretationrestriction}
		Let $\mathcal{T}$ be a TBox and let $\Gamma$ be a sequent. Then $\Gamma$ is satisfiable under $\mathcal{T}$ iff there exists an interpretation where $\Gamma$ is satisfied by some individual and each individual satisfies $T \geq 1$ where $T$ is the  concept assertion associated to~$\mathcal{T}$ and~$\Gamma$.
	\end{lemma}
	\begin{proofsketch}
		'If' is trivial by the above argumentation. 'Only if' can be shown by transforming an arbitrary interpretation satisfying $\Gamma$ under $\mathcal{T}$ into one where all atomic concepts in all individuals and all roles have values in $Z'$ by either keeping them as is if they already were in $Z'$ or taking the closest value in $Z' \setminus Z$ otherwise. One can then show by induction that concept assertions containing only values of $Z$ on their right hand side and subconcepts of $\Gamma$ and $\mathcal{T}$ on their left hand side cannot distinguish between these interpretations and that $T \geq 1$ is satisfied in every individual.
	\end{proofsketch}
	\begin{definition}
		For ${\triangleleft} \in \{ <, \leq \}$ and ${\triangleright} \in \{ >, \geq \}$, we put
		\begin{equation*}
			\mathrel{\blacktriangleleft_{({\triangleleft}, {\triangleright})}} :=
			\begin{cases}
				<, &\text{if }({\triangleleft}, {\triangleright})=(\leq, \geq)\\
				\leq, &\text{otherwise.}
			\end{cases}
		\end{equation*}
	\end{definition}
	\begin{remark}
		The $\mathrel{\blacktriangleleft_{({\triangleleft}, {\triangleright})}}$ operator is used to obtain a non-strict inequality
		if at least one of the inputs is a strictly greater or smaller
		than. This will be useful when
		determining if a pair of inequalities is solvable or not; explicitly, if we have
		$p \leq c$ and $p \geq d$ we cannot find a valid value for $p$ if
		$c < d$ and in all other cases $p \mathrel{\triangleleft} c$ and $p \mathrel{\triangleright} d$ we cannot
		find a valid value for $p$ if $c \leq d$.
	\end{remark}
	\noindent If $T \geq 1$ is the associated concept
	assertion for the TBox $\mathcal{T}$ and a sequent $\Gamma$, this gives us the tableau
	calculus of Table \ref{tab:tableau}.
	The axioms $(\text{Ax } 1), (\text{Ax } 0)$ and $(\text{Ax } p)$ assert that  the truth value of an atomic concept cannot be larger than~$1$, smaller than~$0$ or in an empty interval. The axiom $(\text{Ax } c)$ asserts that  $c \mathrel{\triangleleft} d$ is satisfiable only if this constraint actually holds for the constants $c,d$.
	The $(\sqcap \mathrel{\triangleright}), (\sqcap \mathrel{\triangleleft})$ and $(\lnot \mathrel{\bowtie})$ rules are the usual fuzzy propositional rules.
	The $(\ominus \mathrel{\triangleleft})$ rule asserts that we can add $c$ to both sides of a concept assertion without affecting satisfiability, as long as the assertion is not already unsatisfiable
	by the fact that $C \ominus c$ would have to be smaller than $0$. The $(\ominus \mathrel{\triangleright})$ rule also asserts that we can add $c$ to both sides of the concept assertion. However, this time we have to be careful not to apply this to concept assertions that are always satisfied, i.e.\  assertions saying that the left hand side should be greater than or equal to~$0$. 
	We have to avoid this case, as otherwise these trivial assertions could be transformed into unsatisfiable ones, e.g. the satisfiable concept assertion $0 \ominus 1 \geq 0$ would be transformed into $0 \geq 1$, which is not satisfiable.
	The rule $(\exists R)$ first takes all universal restrictions
	of some role $R \in \mathsf{N_R}$ and one existential restriction $\exists R . C \mathrel{\triangleright} c$. We then filter out all universal restrictions $\exists R . D_j \mathrel{\triangleleft} d_j$ that can be trivially satisfied by finding a value~$e$ for the role~$R$ such that $e \mathrel{\triangleright} c$ but $e \mathrel{\triangleleft} d_j$. This means we only have universal restrictions $\exists R . D_j \mathrel{\triangleleft} d_j$ left with $d_j \mathrel{\blacktriangleleft_{({\triangleleft_j}, {\triangleright})}} c$. We then take these universal restrictions and the existential restriction, remove their $\exists R .$ part, and add 
	$T \geq 1$, while dropping the remaining context $S$, to see if we can create an individual satisfying these constraints.
	\begin{definition}
		\begin{enumerate}[wide]
			\item The \emph{propositional rules} of the calculus are all rules except $(\exists R)$.
			\item As usual, for each rule we take all its possible instances for concrete  constants, concepts, and sets of concept assertions, and say the the rule is \emph{applicable} to a sequent if the sequent matches the premise of an instance of the rule,  satisfying possible side conditions of the rule.
			\item When we say we \emph{apply} a rule to a sequent, we mean that we take the conclusions of a  rule instance whose premise matches the sequent.
			\item A sequent
			is \emph{saturated} if no propositional rules can be applied to it.
		\end{enumerate}
	\end{definition}
	\noindent The idea of a global caching algorithm to construct a tableau is to cache all the labels of its nodes, and whenever a new node with an already encountered label would be created, we instead create an edge to the already existing node with that label, thus obtaining a directed graph, possibly with cycles. As we will see later, this will ensure the algorithm terminates in exponentially many steps.
	\begin{definition}
		\begin{enumerate}[wide]
			\item A \emph{tableau graph} $G$ consists of
			\begin{itemize}[wide]
				\item a directed graph $(V^G,E^G)$ consisting of a set~$V^G$ of \emph{nodes} and a set~$E^G\subseteq V^G\times V^G$ of \emph{edges};
				\item a \emph{root node} $r^G \in V^G$;
				\item and for each node~$v\in V^G$ a \emph{label} $\mathcal{L}^G(v)$, which is a set of concept assertions or $\bot$.
			\end{itemize}
			\item A node $v\in V^G$ in a tableau graph~$G$ is  an \emph{AND-node} if its label $\mathcal{L}^G(v)$ is saturated, and otherwise an \emph{OR-node}. We denote by $A^G$ the set of AND-nodes and by~$O^G$ the set of OR-nodes.
		\end{enumerate}
	\end{definition}
	\begin{definition}
		Let $G$ be a tableau graph and $v \in V^G$ be a node.
		\begin{enumerate}[wide]
			\item A propositional rule has been \emph{applied} to $v$ in $G$ if the labels of the child nodes of~$v$ are exactly the conclusions of an application of an instance of this rule to the label of~$v$.
			\item The rule $(\exists R)$ has been \emph{applied} to $v$ in $G$ if the labels of the child nodes of~$v$ are  exactly the conclusions of all possible applications of this rule (maybe none) to the label of~$v$.
			\item We call $G$ a \emph{tableau} for a sequent $\Gamma$ and $T \geq 1$ if the root node $r^G$ has the label $\mathcal{L}^G(r) = \Gamma \cup \{T \geq 1\}$, a propositional rule has been applied to all OR-nodes and the rule $(\exists R)$ has been applied to all AND-nodes.
		\end{enumerate} 
	\end{definition}
	\noindent The next goal is to prove soundness and completeness of this tableau calculus and then termination and its complexity bound.
	\begin{definition}
		Let $G$ be a tableau for $\Gamma$ and $T \geq 1$.
		\begin{enumerate}[wide]
			\item A \emph{marking} $G_c := (V^{G_c}, E^{G_c})$ of $G$ is a full subgraph of $(V^G, E^G)$ where
			\begin{itemize}[wide]
				\item $r^G \in V^{G_c}$;
				\item for $v \in V^{G_c}, v \in A^G$, we have $w \in V^{G_c}$ for all $(v,w) \in E^G$; and
				\item for $v \in V^{G_c}, v \in O^G$, we have $w \in V^{G_c}$ for some $(v,w) \in E^G$.
			\end{itemize}
			\item A marking is \emph{consistent} if it does not contain a node with label $\bot$.
			\item We call $G$ \emph{open} if there exists a consistent marking, and \emph{closed} otherwise.
		\end{enumerate}
	\end{definition}
	\noindent As we will see later, a consistent marking of a tableau is a postfixpoint of a functional.
	We now show that an open tableau for $\Gamma$ and $T \geq 1$ can be used to construct an interpretation that satisfies $\mathcal{T}$ and where some individual satisfies $\Gamma$. We construct such an interpretation  by first considering $G_c$-saturation paths:
	\begin{definition}
		Let $G_c$ be a marking of a tableau $G$ and $v_0 \in V^{G_c}$. A $G_c$-\emph{saturation path} from $v_0$ is a finite sequence $v_0, v_1, \ldots, v_n$ such that for all $0 \leq i <n$ we have $(v_i, v_{i+1}) \in E^{G_c}$, $\{v_0, \ldots, v_{n-1}\} \subseteq O^G$, and $v_n \in A^G$.
	\end{definition}
	\begin{remark}
		The above definition implies that a node is an AND-node iff it has a $G_c$-saturation path of length~$0$. Furthermore, all nodes have a $G_c$-saturation path, as we would otherwise need to have some infinite path that only visits OR-nodes. However, since all rules that can be applied to an OR-node decrease the syntactic size of concepts in the label, such a path cannot exist.
	\end{remark}
	\begin{definition}
		We call a sequent $\Gamma$ \emph{clashing} if at least one of
		the following holds:
		\begin{enumerate}[wide]
			\item $p \mathrel{\triangleright} c \in \Gamma$ for some $p \in \mathsf{N_C}, c \mathrel{\overline{\triangleright}} 1$ or $p \mathrel{\triangleleft} c \in \Gamma$ for some $p \in \mathsf{N_C}, c \mathrel{\overline{\triangleleft}} 0$;
			\item $p \mathrel{\triangleleft} c \in \Gamma, p \mathrel{\triangleright} d\in \Gamma$ for some $p \in \mathsf{N_C}$, $c \mathrel{\blacktriangleleft_{({\triangleleft}, {\triangleright})}} d$; 
			\item $c \mathrel{\triangleleft} d \in \Gamma$ for some $c \mathrel{\overline{\triangleleft}^\circ} d$.
		\end{enumerate}
	\end{definition}
	We now  prove completeness of the tableau calculus by extracting an interpretation from a consistent marking:
	\begin{theorem}\label{tableau:completeness}
		Let $\Gamma$ be a sequent and $\mathcal{T}$ be a TBox with associated concept assertion $T \geq 1$. If there exists an open tableau $G$ for $\Gamma$ and $T \geq 1$, then $\Gamma$ is $\mathcal T$-satisfiable.
	\end{theorem}
	\begin{proofsketch}
		We directly construct an interpretation $\mathcal{I}$ from a consistent marking $G_c$ of the open tableau $G$ by contracting the marking along $G_c$-saturation paths:
		\begin{enumerate}[wide]
			\item Put $\Delta^\mathcal{I} := A^G \cap V^{G_c}$.
			\item For every $x \in \Delta^\mathcal{I}$, we put the value of atomic concepts as any value satisfying all the concept assertions in its corresponding label. We can do this since the label of an $x \in V^{G_c}$ can never be clashing.
			\item We take for every $x,y \in \Delta^\mathcal{I}$ and role name $R \in \mathsf{N_R}$ all the nodes along $G_c$-saturation paths starting at a child node of $x$ and ending with~$y$ and by investigating their labels determine which node corresponds to which existential restriction in $x$. We then find a value for $R^\mathcal{I} (x,y)$  that  is large enough to satisfy all these existential restrictions but small enough such that all universal restrictions $\exists R . D \mathrel{\triangleleft} d$ are either satisfied by the value of the transition or $D \mathrel{\triangleleft} d$ is part of the label of a child node of $x$ that has a $G_c$-saturation path ending with~$y$.
		\end{enumerate}
		By induction on~$C$, we have $C^\mathcal{I}(x) \mathrel{\bowtie} c$ for every $x \in \Delta^\mathcal{I}$ and every $C \mathrel{\bowtie} c \in \mathcal{L}(y)$ where $y \in V^{G_c}$ has a $G_c$-saturation path ending with $x$.
		Investigating maximal $G_c$-saturation paths and the tableau rules, we notice that $T \geq 1$ must therefore always be satisfied in every $x \in \Delta^\mathcal{I}$ and that $\Gamma$ is satisfied in some $\tau \in \Delta^\mathcal{I}$ where there is a $G_c$-saturation path starting with the root node $r^{G_c}$ and ending at $\tau$.
	\end{proofsketch}
	\noindent Next, we prove soundness:
	\begin{theorem}\label{tableau:soundness}
		The above tableau calculus is sound. That is, if\/~$\Gamma$ is $\mathcal{T}$-satisfiable, then every tableau for~$\Gamma$ and $T \geq 1$ has a consistent marking, where $T \geq 1$ is the concept assertion  associated to $\mathcal{T}$ and $\Gamma$.
	\end{theorem}
	\begin{proofsketch}
		Let~$G$ be a tableau for~$\Gamma$ and $T \geq 1$. The label of the
		root~$r^G$ can never be~$\bot$, and by checking the rules one by one,
		we see that they create some satisfiable conclusion or, in the case
		of the rule $(\exists R)$, only satisfiable conclusions from a
		satisfiable premise.  This shows that for all nodes with satisfiable
		label, there is at least one child node for an OR-node that has a
		satisfiable label and all child nodes of AND-nodes have 
		satisfiable labels. This means that if we start at the root node
		$r^G$ we always have at least one child with a satisfiable label or
		all children have a satisfiable label for AND-nodes, as is required
		to construct~$G_c$.
	\end{proofsketch}
	\begin{algorithm*}[!ht]
		\caption{checking satisfiability in non-expansive fuzzy $\ALC$}
		\label{alg:sat}
		\KwIn{a sequent $\Gamma$ and an associated concept assertion $T \geq 1$ for a TBox $\mathcal{T}$ and $\Gamma$}
		\KwOut{ true if all concept assertions of $\Gamma$ are satisfiable under the TBox $\mathcal{T}$, false otherwise}
		construct a tableau $G$ for $\Gamma$ and $T \geq 1$;\\
		$\operatorname{Unsat} := \emptyset$, $\operatorname{Queue} := \emptyset$;\\
		\If{$\exists v_\bot \in V^G : \mathcal{L}^G(v_\bot)=\{\bot\}$} {
			$\operatorname{Unsat} := \{v_\bot\}$, $\operatorname{Queue} := \{v_\bot\}$;\\
			\While{$\operatorname{Queue} \neq \emptyset$} {
				remove some $w$ from $\operatorname{Queue}$;\\
				\ForAll{$p \in V^G : (p,w) \in E^G$} {
					\If{($p \notin \operatorname{Unsat}$ and $p \in A^G$ or $p\in O^G, \forall (p,q) \in E^G: q \in \operatorname{Unsat}$)} {
						add $p$ to $\operatorname{Unsat}$ and $\operatorname{Queue}$;\\
					}
				}
			}
		}
		\eIf{$r^G \in \operatorname{Unsat}$} {
			\Return{false};\\
		} {
			\Return{true};\\
		}
	\end{algorithm*}
	\noindent The procedure and proof for the $\EXP$ bound and termination is inspired by \cite[Section~5]{Gore2013}. It is a construction of the least fixpoint of a functional calculating unsatisfiable nodes. 
	More specifially, we start with the node with label $\bot$ as the only unsatisfiable node and then apply the functional to the set of unsatisfiable nodes until we reach a fixpoint. 
	A node is then in this least fixpoint iff its label is $\mathcal{T}$-unsatisfiable.
	This entails the following tight complexity bound:
	\begin{theorem}\label{theorem:exptime}
		Algorithm \ref{alg:sat} is an $\EXP$ decision procedure for checking satisfiability of the concept assertions in $\Gamma$ with regards to a TBox $\mathcal{T}$.
	\end{theorem}
	\begin{proof}
		Let $\Gamma$ be a sequent, $T \geq 1$ be the  concept assertion associated to $\mathcal{T}$ and $\Gamma$ and $G$ a tableau for $\Gamma$ and $T \geq 1$.
		\begin{enumerate}[wide]
			\item By investigating the rules and the design of the concept assertion associated to $\mathcal{T}$ and $\Gamma$, one can show that there are at most $2^{O(n^3)}$ possible labels, where $n$ is the syntactic size of $\Gamma$ and $\mathcal{T}$. More specifically, there are $2^{O(n)}$ possible labels arising from applying rules to $\Gamma$, and for every GCI $C \sqsubseteq D$ in $\mathcal{T}$, there are $2^{O(n^2)}$ possible labels that can be obtained from the corresponding conjunct $\sqcup_{z \in Z'}(\lnot C \oplus z) \sqcap (D \oplus (1-z)) \geq 1$  in~$T$. Since we have at most $O(n)$ GCIs, we thus have $2^{O(n^3)}$ possible labels that arise from the TBox, and multiplying  with the $2^{O(n)}$ labels of $\Gamma$, we obtain at most $2^{O(n^3)}$ labels.
			\item We now immediately have that Algorithm \ref{alg:sat} terminates after at most $2^{O(n^3)}$ steps, since there are $2^{O(n^3)}$ nodes in $G$.
			\item The condition $r^G \notin \operatorname{Unsat}$ is equivalent to the existence of a consistent marking for $G$, where $\operatorname{Unsat}$ is the set of unsatisfiable nodes as computed in Algorithm \ref{alg:sat}:
			For $v \in A^G$ we have $v \notin \operatorname{Unsat}$ iff $v$ does not have $\bot$ as its label and for all nodes $w \in V^G$ with $(v,w) \in E^G$ we have $w \notin \operatorname{Unsat}$. Similarly for $v \in O^G$ we have $v \notin \operatorname{Unsat}$ iff there exists a $w \in V^G$ with $(v,w) \in E^G$ and $w \notin \operatorname{Unsat}$. Taking all the nodes that are not in $\operatorname{Unsat}$ then by definition is a consistent marking iff $r^G \notin \operatorname{Unsat}$. This also shows that a consistent marking is a postfixpoint of some functional, as Algorithm \ref{alg:sat} calculates the least fixpoint for $\operatorname{Unsat}$.
		\end{enumerate}
	\end{proof}
	\begin{theorem}\label{theorem:hardness}
		Checking a sequent $\Gamma$ for $\mathcal{T}$-satisfiability is $\EXP$ hard.
	\end{theorem}
	\begin{proof}
		We reduce $\ALC$ satisfiability under a TBox, which is known to be $\EXP$ hard \cite{schild1994terminological}, to $\mathcal{T}$-satisfiability in non-expansive fuzzy $\ALC$: Let $\Gamma$ be a set of concepts and $\mathcal{T}$ be a TBox for $\ALC$. Let $\Gamma' := \{C \geq 1 \mid C \in \Gamma\}$ and $\mathcal{T}' := \mathcal{T}$. If $\Gamma$ is satisfiable under~$\mathcal{T}$ in $\ALC$ by some interpretation $\mathcal{I}$, we obtain an interpretation $\mathcal{I}'$ for non-expansive fuzzy $\ALC$ by putting $\Delta^{\mathcal{I}'} := \Delta^\mathcal{I}$ and 
		\begin{equation*}
			p^{\mathcal{I}'}(x) :=
			\begin{cases}
				1 &\text{if } x \in p^\mathcal{I}\\
				0 &\text{otherwise}
			\end{cases}
		\end{equation*}
		\begin{equation*}
			R^{\mathcal{I}'}(x,y) :=
			\begin{cases}
				1 &\text{if } (x,y) \in R^\mathcal{I}\\
				0 &\text{otherwise}.
			\end{cases}
		\end{equation*}
		We then obviously have $\Gamma'$ is satisfiable under~$\mathcal{T}'$, as the evaluation of concepts is the same in this case as in non-fuzzy $\ALC$. On the other hand, if $\Gamma'$ is satisfiable under $\mathcal{T}'$ by some interpretation $\mathcal{I}'$ in non-expansive fuzzy $\ALC$, then we obtain an interpretation $\mathcal{I}$ for $\ALC$ by $\Delta^\mathcal{I} := \Delta^{\mathcal{I}'}$, $x \in p^\mathcal{I}$ iff $ p^{\mathcal{I}'}(x)>0.5$ and $(x,y) \in R^\mathcal{I}$ iff $ R^{\mathcal{I}'}(x,y)>0.5$.
		One can then prove $C^{\mathcal{I}'}(x) > 0.5$ iff $x \in C^\mathcal{I}$ for all concepts $C$ of regular $\ALC$ by induction.
		This then implies that $\Gamma$ is satisfiable under $\mathcal{T}$.
	\end{proof}
	\begin{remark}
		We can modify Algorithm \ref{alg:sat} as in \cite[Section~6]{Gore2013} to decide and propagate satisfiability or unsatisfiability on the fly when constructing the tableau. After expanding a node, we immediately mark it as unsatisfiable if the node has label $\bot$, if the node is an OR-node and all child nodes are marked as unsatisfiable or if the node is an AND-node and at least one child node is marked as unsatisfiable. On the other hand we mark a node with a non-$\bot$ label as satisfiable if the node is an OR-node and at least one child node is marked as satisfiable or if the node is an AND-node and all child nodes are marked as satisfiable. If a node has been marked in this way, we propagate these results upward, i.e.\ we check these conditions again for all parent nodes and if their status changed we propagate again and so on. Thus we can stop expanding the tableau whenever we can decide if the root node is satisfiable or not instead of constructing the full tableau.
		One can also make propagation an optional step that can be applied instead of expanding the graph as seen in \cite{10.1007/978-3-642-12002-2_9} and \cite{10.1007/978-3-642-14203-1_5}.
	\end{remark}
	\begin{remark}
		For readability, we have so far elided ABoxes from the technical development. As usual, a fuzzy  ABox consists of concept assertion $C(a)\mathrel{\bowtie} c$ and role assertions $R(a,b)\mathrel{\bowtie} c$, with the expected semantics, where $a,b\in\mathsf{N}_i$ for a dedicated name space $\mathsf{N}_i$ of \emph{individuals}, which denote elements of the domain.  
		The calculus is extended to handle ABoxes in a straightforward manner by just initializing the run of the tableau procedure with additional root nodes for the individuals mentioned in the ABox, containing all concepts asserted for the respective individual and connected by edges reflecting the role assertions. The complexity remains unaffected.
	\end{remark}
	
	\section{Conclusion}
	We have introduced the description logic \emph{non-expansive fuzzy
		$\ALC$}, which lies between the \L{}ukasiewicz and Zadeh variants of
	fuzzy $\ALC$, notably featuring constant shift operators. In
	particular in connection with general TBoxes, expressivity is markedly
	increased in comparison to Zadeh fuzzy $\ALC$; nevertheless, we have
	shown that the complexity of reasoning remains $\EXP$, the same as for
	two-valued $\ALC$ in sharp contrast to the undecidability encountered
	in the case of \L{}ukasiewicz fuzzy
	$\ALC$~\cite{BaaderPenaloza11}.
	Future work will partly concern coverage of additional features, in
	particular transitive roles, inverses, and nominals. These have been
	integrated fairly smoothly into the less expressive Zadeh variant in
	earlier work~\cite{StoilosEA07,StoilosEA14}; for non-expansive fuzzy
	$\ALC$, the degree of adaptation to the tableau system required to
	accommodate these features remains to be explored.
	\section*{Acknowledgements}
	
	This work is supported by the Deutsche Forschungsgemeinschaft (DFG, German Research Foundation) -- project number 531706730.
	
	We extend recognition to Philipp Hermes, who developed a preliminary unlabelled tableau calculus for non-expansive fuzzy $\ALC$ without TBoxes within his  bachelor thesis supervised by the second author~\cite{Hermes23}.
	
	\bibliographystyle{named}
	\bibliography{tableaux}

\appendix

	\section{Appendix}
	We collect some of the more technical and lengthy details of the proofs. We start with the following statement from the proof of Lemma \ref{lemma:interpretationrestriction}.
	\begin{lemma}\label{lemma:interpretationrestriction2}
		Let $\mathcal{T}$ be a TBox and $\Gamma$ be a sequent. Let $Z$ be the additive group of $1$ and the constants appearing in $\Gamma$ and $\mathcal{T}$ intersected with the unit interval, $\epsilon := \frac{1}{2}\min (Z \setminus \{0\})$ and $Z' := Z \cup \{z+\epsilon \mid z \in Z \setminus\{1\}\}$. Then $\Gamma$ is satisfiable under $\mathcal{T}$ iff there exists an interpretation $\mathcal{I}$ and individual $x \in \Delta^\mathcal{I}$ where $\Gamma$ is satisfied in $x$ and for every $A \sqsubseteq B \in \mathcal{T}$ and every individual $y \in \Delta^\mathcal{I}$ we have some constant $z \in Z'$ with $((\lnot A \oplus z) \sqcap (B \oplus (1-z)))^\mathcal{I}(y) \geq 1$.
	\end{lemma}
	\begin{proof}
		If there exists such an interpretation then clearly $\Gamma$ is
		satisfiable under $\mathcal{T}$ with this interpretation as
		witness. On the other hand let $\mathcal{I}'$ be an interpretation
		that satisfies $\Gamma$ under $\mathcal{T}$. We construct
		$\mathcal{I}$ in the following way: Put
		$\Delta^\mathcal{I} := \Delta^{\mathcal{I}'}$. For every individual
		$x \in \Delta^\mathcal{I}$ put
		$p^\mathcal{I}(x) := p^{\mathcal{I}'}(x)$ if
		$p^{\mathcal{I}'}(x) \in Z$ and otherwise put
		$p^\mathcal{I}(x) := \max\{z\oplus\epsilon \mid z \in Z, z <
		p^{\mathcal{I}'}(x)\}$. For every role $R \in \mathsf{N_R}$ and individuals
		$x,y \in \Delta^\mathcal{I}$ we put
		$R^\mathcal{I}(x,y) := R^{\mathcal{I}'}(x,y)$ if
		$R^{\mathcal{I}'}(x,y) \in Z$ and
		$R^\mathcal{I}(x,y) := \max\{z\oplus\epsilon \mid z \in Z, z <
		R^{\mathcal{I}'}(x,y)\}$ otherwise. We prove by induction on~$C$
		that for every subconcept $C$ of $\Gamma$ and $\mathcal{T}$, every individual
		$x\in \Delta^\mathcal{I}$, and every constant $z \in Z$,
		$C^{\mathcal{I}'}(x) \mathrel{\bowtie} z$ iff $C^\mathcal{I}(x) \mathrel{\bowtie} z$,
		and if $\mathrel{\bowtie} \in \{ \geq, \leq\}$, then this even holds for all
		$z \in Z'$ with $z \neq C^\mathcal{I}(x)$.
		\begin{itemize}[wide]
			\item $C=c$  a constant: Trivial.
			\item $C=p$ an atomic concept: We have by definition that either
			$p^\mathcal{I}(x) = p^{\mathcal{I}'}(x)$, in which case the claim
			is trivial, or
			$p^\mathcal{I}(x) = \max\{z+\epsilon \mid z \in Z, z <
			p^{\mathcal{I}'}(x)\}$. In the latter case we clearly have for all
			$z \in Z$ that $z < p^{\mathcal{I}'}(x)$ iff
			$z < p^\mathcal{I}(x)$. The other inequalities directly follow.
																					\item Induction step: 
			\begin{itemize}[wide]
				\item $C=\lnot D$: We then have $D^{\mathcal{I}'}(x) \mathrel{\bowtie} 1-z$ iff $D^\mathcal{I}(x) \mathrel{\bowtie} 1-z$ by the induction hypothesis and the claim immediately follows.
				\item $C=D \ominus c$: If $(0, <) \neq (z, \mathrel{\bowtie}) \neq (0, \geq)$ and $(1, >) \neq (z, \mathrel{\bowtie}) \neq (1, \leq)$ we then have $D^{\mathcal{I}'}(x) \mathrel{\bowtie} z+c$ iff $D^\mathcal{I}(x) \mathrel{\bowtie} z+c$ by the induction hypothesis and the claim immediately follows. If $(0, <) = (z, \mathrel{\bowtie})$ or $(1, >) = (z, \mathrel{\bowtie})$ then both assertions are never satisfiable and if $(z, \mathrel{\bowtie}) = (0, \geq)$ or $(1, \leq) = (z, \mathrel{\bowtie})$ then the assertions are trivially always satisfied.
				\item $C= D_1 \sqcap D_2$: Using the induction hypothesis we have $D_1^{\mathcal{I}'}(x) \mathrel{\bowtie} z$ iff $D_1^\mathcal{I}(x) \mathrel{\bowtie} z$ and $D_2^{\mathcal{I}'}(x) \mathrel{\bowtie} z$ iff $D_2^\mathcal{I}(x) \mathrel{\bowtie} z$. If $\mathrel{\bowtie} = \mathrel{\triangleright}$ then we have $C^{\mathcal{I}'}(x) \mathrel{\triangleright} z$ iff $D_1^{\mathcal{I}'}(x) \mathrel{\triangleright} z, D_2^{\mathcal{I}'}(x) \mathrel{\triangleright} z$. By the induction hypothesis we then obtain $C^{\mathcal{I}'}(x) \mathrel{\triangleright} z$ iff $D_1^{\mathcal{I}}(x) \mathrel{\triangleright} z, D_2^{\mathcal{I}}(x) \mathrel{\triangleright} z$ which gives us $C^{\mathcal{I}'}(x) \mathrel{\triangleright} z \leftrightarrow C^{\mathcal{I}}(x) \mathrel{\triangleright} z$. On the other hand if $\mathrel{\bowtie} = \mathrel{\triangleleft}$ then we have $C^{\mathcal{I}'}(x) \mathrel{\triangleleft} z$ iff $D_1^{\mathcal{I}'}(x) \mathrel{\triangleleft} z$ or $D_2^{\mathcal{I}'}(x) \mathrel{\triangleleft} z$. Using the induction hypothesis on the corresponding case then immediately proves the claim.
				\item $C = \exists R . D$: If $\mathrel{\bowtie} = \mathrel{\triangleright}$ then $(\exists R . D)^{\mathcal{I}'}(x) \mathrel{\triangleright} z$ iff there exists $y \in \Delta^{\mathcal{I}'}$ with $R^{\mathcal{I}'}(x,y) \mathrel{\triangleright} z$ and $D^{\mathcal{I}'}(y) \mathrel{\triangleright} z$. The first one is by definition true iff $R^{\mathcal{I}}(x,y) \mathrel{\triangleright} z$ and the second one is true iff $D^{\mathcal{I}}(y) \mathrel{\triangleright} z$ by the induction hypothesis. Finally if $\mathrel{\bowtie} = \mathrel{\triangleleft}$ then for all $y \in \Delta^{\mathcal{I}'}$ we have $R^{\mathcal{I}'}(x,y) \mathrel{\triangleleft} z$ or $D^{\mathcal{I}'}(y) \mathrel{\triangleleft} z$ iff $(\exists R . D)^{\mathcal{I}'}(x) \mathrel{\triangleleft} z$. The first one is by defintion equivalent to $R^{\mathcal{I}}(x,y) \mathrel{\triangleleft} z$ and the second one by the induction hypothesis to $D^{\mathcal{I}}(y) \mathrel{\triangleleft} z$ which concludes the case.
			\end{itemize}
		\end{itemize}
		Using this we immediately obtain a constant $z \in Z'$ for every individual $x \in \Delta^\mathcal{I}$ and GCI $A \sqsubseteq B \in \mathcal{T}$ with $((\lnot A \oplus z) \sqcap (B \oplus (1-z)))^\mathcal{I}(x) \geq 1$ by for example taking $z = A^\mathcal{I}(x)$. Now let $x$ be the individual such that for all $C \mathrel{\bowtie} c \in \Gamma$ we have $C^{\mathcal{I}'}(x) \mathrel{\bowtie} c$. We then also have $C^{\mathcal{I}}(x) \mathrel{\bowtie} c$ by the induction and the fact that $c \in Z$.
	\end{proof}
														\noindent We have the following lemmas for the proof of Theorem \ref{tableau:completeness}:
	\begin{lemma}\label{lem:satpath}
		Let $G_c$ be a consistent marking for a tableau $G$ and $v_0, \ldots, v_n$ be a $G_c$-saturation path of $v_0$. Then:
		\begin{enumerate}[wide]
			\item all assertions of the form $p \mathrel{\bowtie} c$, $d \mathrel{\bowtie} c$, $\exists R . C \mathrel{\bowtie} c$ 
			of the label of each $v_i$ are in the label of $v_n$, where $p \in \mathsf{N_C}$.
			\item the label of each $v_i$ does not clash.
			\item the set formed by taking the union of the labels $v_0, v_1, \ldots, v_n$ does not clash.
		\end{enumerate}
	\end{lemma}
	\begin{proof}
		\begin{enumerate}[wide]
			\item This is trivial, as no rule operates on these assertions apart from the axioms which are disqualified because of it being a consistent marking and the $(\exists R)$ rule, which cannot be in the middle of such a path, as it would be an AND-node.
			\item If we have clashing assertions in some node $v$, there would have to be an axiom applied in $G$ to all successors of the corresponding node of $G$ before any other AND-node 
						can happen. This means that the AND-node $v_n$ would be the conclusion of an axiom, i.e. it would have a $\bot$ label, which violates the consistent marking condition. Contradiction.
			\item We take the union of the labels of $v_0, \ldots, v_n$. If there are clashing assertions, we have the following possibilities: either $p \mathrel{\triangleright} c$ with $p \in \mathsf{N_C}, c \mathrel{\overline{\triangleright}} 1$ or $p \mathrel{\triangleleft} c$ with $p \in \mathsf{N_C}, c \mathrel{\overline{\triangleleft}} 0$ is in the union, but then it would have been part of some $v_i$ which contradicts statement 2. The same is true for $c \mathrel{\triangleleft} d$ with $c \mathrel{\overline{\triangleleft}^\circ} d$, as this assertion would be in $v_n$ by statement 1. and by statement 2. this cannot be the case. Finally if we have the case $p \mathrel{\triangleleft} c, p \mathrel{\triangleright} d$ with $p \in \mathsf{N_C}$, $d \blacktriangleleft_{(\mathrel{\triangleleft}, \mathrel{\triangleright})} c$, we once again use statement 1. to assert that both of these must be in $v_n$ and as such by statement 2. this is a contradiction. 		\end{enumerate}
	\end{proof}
	\begin{definition}
		Let $G$ be a tableau graph and $v \in V^G$ be a node. Let $\exists R . C \mathrel{\triangleright} c \in \mathcal{L}^G(v)$. If the exists rule has been applied to $v$ in $G$ we call the conclusion of the instance of the exists rule, where $\exists R . C \mathrel{\triangleright} c$ is the relevant existential restriction, the \emph{associated conclusion} of $\exists R . C \mathrel{\triangleright} c$ in $v$. We will also write $\operatorname{AC}(\exists R . C \mathrel{\triangleright} c, v)$ for the associated conclusion.
	\end{definition}
	\begin{definition}\label{definition:extractedInterpretation}
		Let $G_c$ be a consistent marking for a tableau $G$. Then the interpretation $\mathcal{I}$ extracted from $G_c$, which we outlined in the proof of Theorem \ref{tableau:completeness}, is defined by the following procedure:
		\begin{enumerate}[wide]
			\item Put $\Delta^\mathcal{I} := A^G \cap V^{G_c}$.
			\item For $x \in \Delta^\mathcal{I}$, take all nodes along a $G_c$-saturation paths ending with $x$ as $\operatorname{SP}(x) := \{v_i \mid 0 \leq i \leq n, (v_0, \ldots, v_n=x) \text{ is a } G_c\text{-saturation path}\}$ and define $Y(x):=\bigcup_{v \in \operatorname{SP}(x)} \mathcal{L}^G(v)$ as the union of all labels of $G_c$-saturation paths ending with $x$. For every atomic concept $p$ we can then find a value $c \in [0,1]$ such that for every $p \mathrel{\bowtie} d \in Y(x)$ we have $c \mathrel{\bowtie} d$ is true. Define $p^\mathcal{I}(x) := c$. We can do this since the concept assertions in $Y(x)$ are not clashing.
			\item Let $x \in \Delta^\mathcal{I}$, $y \in \Delta^\mathcal{I}$ and $R$ be a role name. Define $\operatorname{EX}_{R} (x,y) \coloneqq \bigcup_{w \in \operatorname{SP}(y)} \{\exists R . D \mathrel{\triangleright} d \in \mathcal{L}^G(x) \mid (x,w) \in E^{G_c}, \mathcal{L}^G(w) = \operatorname{AC}(\exists R . D \mathrel{\triangleright} d, x)\}$ as the concept assertions $\exists R . D \mathrel{\triangleright} d$ in $Y(x)$ such that $D \mathrel{\triangleright} d$ is in $Y(y)$ and the associated conclusion in $x$ is the label of a node on a $G_c$-saturation path that ends at $y$.
			Furthermore, let $\operatorname{FA}_{R}(x):=\{\exists R . D \mathrel{\triangleleft} d \in Y(x)\}$ be the set of universal restrictions for $R$ in $x$. Find $e$ big enough such that for all $(\exists R . D \mathrel{\triangleright} d) \in \operatorname{EX}_{R}(x,y)$ we have $e \mathrel{\triangleright} d$, but small enough such that for all $(\exists R . D \mathrel{\triangleleft} d) \in \operatorname{FA}_{R} (x)$ we have $e \mathrel{\triangleleft} d$ or $(D \mathrel{\triangleleft} d) \in Y(y)$. Put $R ^\mathcal{I} (x,y) := e$. 		\end{enumerate}
	\end{definition}
	\begin{remark}
		We can always do step (3) of the construction in Definition \ref{definition:extractedInterpretation}: if we have $\exists R . D \mathrel{\triangleleft} d \in Y(x)$ then we also have $\exists R . D \mathrel{\triangleleft} d \in \mathcal{L}^G(x)$ and when we apply an exists rule to some $\exists R . F \mathrel{\triangleright} f \in \mathcal{L}^G(x)$ such that the conclusion $S$ is the label of some $w$ in $\operatorname{SP}(y)$ then $S \subseteq Y(y)$. If we have $d \blacktriangleleft_{(\mathrel{\triangleleft}, \mathrel{\triangleright})} f$, then $D \mathrel{\triangleleft} d$ is part of the conclusion, so we have $D \mathrel{\triangleleft} d \in Y(y)$.
		If we have $d \mathrel{\overline{\blacktriangleleft_{(\mathrel{\triangleleft}, \mathrel{\triangleright})}}^\circ} f$ for all $(\exists R . F \mathrel{\triangleright} f) \in \operatorname{EX}_{R}(x,y)$, then we can obviously find $e$ such that $e \mathrel{\triangleleft} d$ but for all $(\exists R . F \mathrel{\triangleright} f) \in \operatorname{EX}_{R}(x,y)$ we have $e \mathrel{\triangleright} f$.
		Doing this for all (finitely many) such assertions and taking the smallest such $e$ proves this claim.
	\end{remark}
	\begin{lemma}\label{lem:rulesapp}
		Let $\mathcal{I}$ be the interpretation extracted from a consistent marking $G_c$ for a tableau $G$ and $C \mathrel{\bowtie} c \in Y(x)$ for some $x \in \Delta^\mathcal{I}$. Then:
		\begin{enumerate}[wide]
			\item If there are OR-rules of the tableau applicable to $C \mathrel{\bowtie} c$ then at least one of the conclusions of an applicable rule is also in $Y(x)$.
			\item There is at most one OR-rule applicable to $C \mathrel{\bowtie} c$.
			\item If $C$ is not a constant, an atomic concept or $C = \exists R_l . D$ and $C \mathrel{\bowtie} c$ is not $D \ominus d \mathrel{\triangleright} c$ with $c \mathrel{\triangleright^\circ} 0$ then an OR-rule is always applicable to $C \mathrel{\bowtie} c$.
			\item If there are AND-rules of the tableau applicable to $C \mathrel{\bowtie} c = \exists R . D \mathrel{\triangleright} c$ then all of its conclusions are in $Y(y)$ and $R^\mathcal{I} (x,y) \mathrel{\triangleright} c$ for some $y \in \Delta^\mathcal{I}$.
		\end{enumerate}
	\end{lemma}
	\begin{proof}
		\begin{enumerate}[wide]
			\item Because of $C \mathrel{\bowtie} c \in Y(x)$ there exists a $G_c$-saturation path $(v_0, \ldots, v_n=x)$ with $C \mathrel{\bowtie} c \in \mathcal{L}^G(v_0)$. A $G_c$-saturation path either stops at an AND-node without successors, in which case no rule is applicable at all, or an AND-node where the exists rule is being applied, in which case no OR-rule is allowed to be applicable because of the saturation condition of the exists rule. This means that if there is any OR-rule applicable to $C \mathrel{\bowtie} c$ an OR-rule must have been applied to it along the $G_c$-saturation path, so we have some $0 \leq i \leq n$ where one of the conclusions of that rule is in the label $\mathcal{L}^G(v_i) \subseteq Y(x)$.
			\item This is true by the design of the grammar and the rules: Since we are in a consistent marking, the axioms cannot be applicable to any concept assertions in $Y(x)$, and by design the non-axiom OR-rules do not overlap.
			\item This is also trivially true by design of the grammar and the tableau rules.
			\item If the exists rule is applicable to $C \mathrel{\bowtie} c = \exists R . D \mathrel{\triangleright} c$, there has to be some $w \in G_c$ where $(x, w) \in E^{G_c}$, $D \mathrel{\triangleright} c \in \mathcal{L}^G(w)$ and for all $\exists R . F \mathrel{\triangleleft} f \in \operatorname{FA}(x)$ we have $F \mathrel{\triangleleft} f \in \mathcal{L}^G(w)$ or $f \mathrel{\triangleleft} c$. Since any node is part of at least one $G_c$-saturation path, this means that by construction we have to have $y \in \Delta^\mathcal{I}$ with $D \mathrel{\triangleright} c \in Y(y)$ and $R^\mathcal{I} (x,y) \mathrel{\triangleright} c$.
		\end{enumerate}
	\end{proof}
	\begin{lemma}\label{lem:CxSatInx}
		Let $\mathcal{I}$ be the interpretation extracted from a consistent marking $G_c$ for a tableau $G$. For every $x \in \Delta^\mathcal{I}$ and every $C \mathrel{\bowtie} c \in Y(x)$ we have $C^\mathcal{I}(x) \mathrel{\bowtie} c$.
	\end{lemma}
	\begin{proof}
		We prove this by induction:
		\begin{itemize}[wide]
			\item Induction base: Let $x \in \Delta^\mathcal{I}$ and $C=d$. Then
			$d \mathrel{\bowtie} c \in Y(x)$ means $d \mathrel{\bowtie} c \in \mathcal{L}^G(v)$
			for some $v \in \operatorname{SP}(x)$. If this was not
			satisfiable, then we would have a clashing assertion in
			$v$. Contradiction by Lemma \ref{lem:satpath} (2). Now let
			$C=p$. Then $p^\mathcal{I}(x) \mathrel{\bowtie} c \in Y(x)$ is satisfied by
			construction.
									\item Induction step:
			Let $x \in \Delta^\mathcal{I}$ and $C \mathrel{\bowtie} c \in Y(x)$. We are trying to prove $C^\mathcal{I}(x) \mathrel{\bowtie} c$. We then have some $v \in \operatorname{SP}(x)$ with $C \mathrel{\bowtie} c \in \mathcal{L}^G(v)$.
			\begin{itemize}[wide]
				\item By Lemma \ref{lem:rulesapp} (1)-(3) we have that for all cases where $C$ is neither a constant, atomic concept or $C = \exists R . D$ and $C \mathrel{\bowtie} c$ is not $D \ominus d \mathrel{\triangleright} c$ with $c \mathrel{\triangleright^\circ} 0$ that either the $(\sqcap \mathrel{\triangleright}), (\sqcap \mathrel{\triangleleft}), (\lnot \mathrel{\bowtie}), (\ominus, \mathrel{\triangleleft})$ or $(\ominus, \mathrel{\triangleright})$ rules were applied to $C \mathrel{\bowtie} c$ along a $G_c$-saturation path and that at least one of its conclusions are also in $Y(x)$. This means we can immediately conclude the cases of $C = \lnot D$, $C =  D \ominus d$ and $C = D \sqcap E$ by using the induction hypothesis on the conclusion that is in $Y(x)$. The case where $C \mathrel{\bowtie} c$ is $D \ominus d \mathrel{\triangleright} c$ with $c \mathrel{\triangleright^\circ} 0$ are trivially satisfied in any interpretation and individual.
				\item We also have by Lemma \ref{lem:rulesapp} (4) that if $C \mathrel{\bowtie} c = \exists R . D \mathrel{\triangleright} c$ then we have some $y \in \Delta^\mathcal{I}$ with $D \mathrel{\triangleright} c \in Y(y)$ and $R^\mathcal{I} (x,y) \mathrel{\triangleright} c$. Using the induction hypothesis again (this time on $y$ and  $D \mathrel{\triangleright} c$) we can immediately conclude this case.
				\item Case $C \mathrel{\bowtie} c = \exists R . D \mathrel{\triangleleft} c$:
				We have to check that for all $y \in \Delta^\mathcal{I}$ we have $R^\mathcal{I} (x,y) \mathrel{\triangleleft} c$ or $D^\mathcal{I}(y) \mathrel{\triangleleft} c$.
				Now let $y \in \Delta^\mathcal{I}$. Then by construction, if we do not have any $\exists R . F \mathrel{\triangleright} f \in Y(x)$ with $c \blacktriangleleft_{(\mathrel{\triangleleft}, \mathrel{\triangleright})} f$, we would have $R^\mathcal{I} (x,y) \mathrel{\triangleleft} c$ and would be done. On the other hand if we have $R^\mathcal{I} (x,y) \mathrel{\overline{\triangleleft}^\circ} c$, that means we have $\exists R . F \mathrel{\triangleright} f \in \operatorname{EX}_{R} (x,y)$ such that $f \mathrel{\overline{\triangleleft}^\circ} c$. We then have $\exists R . F \mathrel{\triangleright} f \in \mathcal{L}^G(x)$ and $\exists R . D \mathrel{\triangleleft} c \in \mathcal{L}^G(x)$. This means that for some $w_k$ along a $G_c$-saturation path $(w_0, \ldots, w_n=y)$ of the consistent marking and thus the tableau we have $F \mathrel{\triangleright} f \in \mathcal{L}^G(w_k)$ and $\exists R . F \mathrel{\triangleright} f \in \operatorname{EX}_{R}(x, w_k)$. Looking at the $(\exists R)$ rule, we then also have to have $D \mathrel{\triangleleft} c \in \mathcal{L}^G(w_k)$ and as such $D \mathrel{\triangleleft} c \in Y(y)$. Using the induction hypothesis on $y$ and $D \mathrel{\triangleleft} c$ then gives us $D^\mathcal{I}(y) \mathrel{\triangleleft} c$.
			\end{itemize}
		\end{itemize}
	\end{proof}
				Next we have the following claim in Theorem \ref{tableau:soundness}:
	\begin{lemma}\label{lem:sound}
		Each non-axiom tableau rule creates a (or all in the case of the exists rule) $\mathcal{T}$-satisfiable conclusion from a $\mathcal{T}$-satisfiable premiss.
	\end{lemma}
	\begin{proof}
		\begin{itemize}[wide]
			\item Case $(\sqcap \mathrel{\triangleright})$: If we have that $S, D \sqcap E \mathrel{\triangleright} d$ is $\mathcal{T}$-satisfiable, then that means there must be an interpretation $\mathcal{I}$ and an individual $x \in \Delta^\mathcal{I}$, where all concept assertions of $S$ and $D \sqcap E \mathrel{\triangleright} d$ are satisfied by $x$ under the TBox $\mathcal{T}$. However if $(D \sqcap E)^\mathcal{I}(x) \mathrel{\triangleright} d$ then by the very definition of $\sqcap$ both $D^\mathcal{I}(x) \mathrel{\triangleright} d$ and $E^\mathcal{I}(x) \mathrel{\triangleright} d$. Thus the conclusion of the rule is $\mathcal{T}$-satisfiable in the same interpretation and for the same individual.
			\item Case $(\sqcap \mathrel{\triangleleft})$: If we have that $S, D \sqcap E \mathrel{\triangleleft} d$ is $\mathcal{T}$-satisfiable, then that means there must be an interpretation $\mathcal{I}$ and an individual $x \in \Delta^\mathcal{I}$, where all concept assertions of $S$ and $D \sqcap E \mathrel{\triangleleft} d$ are satisfied by $x$ under the TBox $\mathcal{T}$. However if $(D \sqcap E)^\mathcal{I}(x) \mathrel{\triangleleft} d$ then by the very definition of $\sqcap$ we have $D^\mathcal{I}(x) \mathrel{\triangleleft} d$ or $E^\mathcal{I}(x) \mathrel{\triangleleft} d$ or both. Thus at least one of the conclusions of the rule is also $\mathcal{T}$-satisfiable in the same interpretation and for the same individual.
			\item Case $(\lnot \mathrel{\bowtie})$: If we have that $S, \lnot D \mathrel{\bowtie} d$ is $\mathcal{T}$-satisfiable, then that means there must be an interpretation $\mathcal{I}$ and an individual $x \in \Delta^\mathcal{I}$, where all concept assertions of $S$ and $\lnot D \mathrel{\bowtie} d$ are satisfied by $x$ under the TBox $\mathcal{T}$. However if $(\lnot D)^\mathcal{I}(x) \mathrel{\bowtie} d$ then by the very definition of $\lnot$ we have $D^\mathcal{I}(x) \mathrel{\bowtie^\circ} 1-d$. Thus the conclusion of the rule is $\mathcal{T}$-satisfiable in the same interpretation and for the same individual.
			\item Case $(\ominus \mathrel{\bowtie})$: If we have that $S, D \ominus d \mathrel{\bowtie} e$ is $\mathcal{T}$-satisfiable, then that means there must be an interpretation $\mathcal{I}$ and an individual $x \in \Delta^\mathcal{I}$, where all concept assertions of $S$ and $D \ominus d \mathrel{\bowtie} e$ are satisfied by $x$ under the TBox $\mathcal{T}$.
			We then have by the definition of $\ominus$ that $(D \ominus d)^\mathcal{I}(x) \mathrel{\bowtie} e \leftrightarrow \max(D^\mathcal{I}(x)-d,0) \mathrel{\bowtie} e$.
			\begin{itemize}[wide]
				\item Case $\mathrel{\bowtie} = \mathrel{\triangleleft}$: Then $e \mathrel{\triangleleft^\circ} 0$ otherwise the concept assertion would never be $\mathcal{T}$-satisfiable. We also trivially have $(D \ominus d)^\mathcal{I}(x) \geq D^\mathcal{I}(x) - d$ and as such $D^\mathcal{I}(x) \mathrel{\triangleleft} e + d$.
				\item Case $\mathrel{\bowtie} = \mathrel{\triangleright}$: Because of the condition $e \mathrel{\overline{\triangleright}} 0$ we have to have $D^\mathcal{I}(x) > d$ in order for the assertion to be $\mathcal{T}$-satisfiable and the $(\ominus \mathrel{\triangleright})$ rule to be applicable. Then we have $(D \ominus d)^\mathcal{I}(x) = D^\mathcal{I}(x) - d$ and as such $D^\mathcal{I}(x) \mathrel{\triangleright} e + d$.
			\end{itemize}
			Thus the conclusions of the rules are $\mathcal{T}$-satisfiable in the same interpretation and for the same individual.
			\item Case axioms: If any axiom rule is applicable, then the premiss was also not $\mathcal{T}$-satisfiable. Contradiction.
			\item Case $(\exists R . \mathrel{\triangleright})$: If we have that $S, \{ \exists R . E_j \mathrel{\triangleleft}_j  e_j \mid 1 \leq j \leq n\}, \exists R . D \mathrel{\triangleright} d$ is $\mathcal{T}$-satisfiable and $d \mathrel{\overline{\triangleright}} 0$, that means there must be an interpretation $\mathcal{I}$ and an individual $x \in \Delta^\mathcal{I}$, where all concept assertions of $S$ and $\{ \exists R . E_j \mathrel{\triangleleft}_j  e_j \mid 1 \leq j \leq n\}, \exists R . D \mathrel{\triangleright} d$ are satisfied by $x$ under the TBox $\mathcal{T}$. If we have for all $1 \leq j \leq n$ that $(\exists R . E_j)^\mathcal{I}(x) \mathrel{\triangleleft}_j e_j$ and $(\exists R . D)^\mathcal{I}(x) \mathrel{\triangleright} d$, that means there must be an individual $y \in \Delta^\mathcal{I}$ with $R^\mathcal{I} (x,y) \mathrel{\triangleright} d$, $D^\mathcal{I}(y) \mathrel{\triangleright} d$ and $E_j^\mathcal{I}(y) \mathrel{\triangleleft}_j e_j$ for at least all $j \in \{1, \ldots, n\}$ with $e_j \blacktriangleleft_{(\mathrel{\triangleleft}_j, \mathrel{\triangleright})} d$ under the TBox $\mathcal{T}$. The other universal restrictions can be trivially satisfied by the value of $R^\mathcal{I} (x,y)$. This is exactly the conclusion of the rule, so the conclusions also have to be $\mathcal{T}$-satisfiable.
		\end{itemize}
	\end{proof}
				We also had the following claim in the proof of Theorem \ref{theorem:exptime}:
	\begin{lemma}\label{lemma:labelsEstimate}
		Let $\Gamma$ be a sequent, $\mathcal{T}$ a TBox and $T \geq 1$ the associated concept assertion to $\mathcal{T}$ and $\Gamma$.
		Then there are at most $2^{O(n^3)}$ possible labels for a tableau of $\Gamma$ and $T \geq 1$, where $n$ is the syntactic size of $\Gamma$ and a TBox $\mathcal{T}$.
	\end{lemma}
	\begin{proof}
		When constructing $G$, we have that for any subconcept of a concept assertion in the current node there can only be one possible comparison operator and constant on the right hand side associated with this subconcept (not counting instances of the same subconcept in another place), as the propositional rules always simplify a concept when changing the comparison operator or constant and for each concept assertion there is at most one propositional rule applicable at all times. 
		The exists rules even keep their comparison operator and constants, so they do not give rise to any other possible comparison operators and constants associated with the subconcept. 
		The only possible way to have some subconcept associated with multiple pairs of comparators and constants is if the subconcept shows up multiple times in the concept. 
		However this still leaves us with at most $2^{O(m)}$ possible labels made up of just subconcept assertions from $\Gamma$, where $m$ is the syntactic size of $\Gamma$. We next have to deal with the possible labels of subconcept assertions of the TBox: 
		Let $n$ be the syntactic size of $\Gamma$ and $\mathcal{T}$. Then there are $O(n)$ constants, each having a numerator and denominator of binary length $O(n)$. The least common multiple of their denominators then has binary length $O(n^2)$ which gives us that the size of the additive group $Z$ of $1$ and constants in $\Gamma$ and $\mathcal{T}$ is of size $2^{O(n^2)}$. Let $\epsilon := \frac{1}{2} \min Z \setminus \{0\}$ then $Z' := Z \cup \{z+\epsilon \mid z \in Z \setminus \{1\}\}$ is also of size $2^{O(n^2)}$. Looking at the rules of the tableau, for each GCI $C \sqsubseteq D$ a label either has some subconcepts of $(\lnot C \oplus z) \sqcap (D \oplus (1-z)) \geq 1$ for one $z \in Z'$ or for a subset $U \subseteq Z'$ the concept assertion $\sqcup_{z \in U}(\lnot C \oplus z) \sqcap (D \oplus (1-z)) \geq 1$. We can however further restrict the possible choices of $U$ by investigating the $(\sqcap \mathrel{\triangleleft})$ rule and noticing that we only have subsets $U_z=\{z\}$ for $z \in Z'$ and $U^c=\{z \mid z \in Z', z \geq c\}$ for constants $c \in Z'$, where we supposed without loss of generality that the disjunction in the concept assertion $\sqcup_{z \in Z'}(\lnot C \oplus z) \sqcap (D \oplus (1-z)) \geq 1$ was ordered based on the ordering of $Z'$. Both of these cases yield $2^{O(n^2)}$ possible labels resulting from the GCI $C \sqsubseteq D$. This then leaves us with $2^{O(n^3)}$ possible labels for all GCIs of $\mathcal{T}$ and multiplying with the $2^{O(m)} \leq 2^{O(n)}$ possible labels resulting from $\Gamma$ we have that there are at most $2^{O(n^3)}$ possible labels in a tableau.
	\end{proof}
	Finally, we had the following claim in the proof of Theorem \ref{theorem:hardness}:
	\begin{lemma}
		Let $\mathcal{I}$ and $\mathcal{I}'$ be the interpretations from the proof of Theorem \ref{theorem:hardness}. Then we have $C^{\mathcal{I}'}(x) > 0.5 \leftrightarrow x \in C^\mathcal{I}$ for all concepts $C$ of regular $\ALC$.
	\end{lemma}
	\begin{proof}
		We prove this by induction:
		\begin{itemize}[wide]
			\item Induction base: The case for atoms is clear by the definition. The case for constants is trivial.
			\item Case $C = D \sqcap E$: We have $C^{\mathcal{I}'}(x) > 0.5 \leftrightarrow D^{\mathcal{I}'}(x) > 0.5, E^{\mathcal{I}'}(x) > 0.5$. Using the induction hypothesis we have $C^{\mathcal{I}'}(x) > 0.5 \leftrightarrow x \in D^\mathcal{I}, x \in E^\mathcal{I}$ and as such $C^{\mathcal{I}'}(x) > 0.5 \leftrightarrow x \in C^\mathcal{I}$.
			\item Case $C = \lnot D$: We have $C^{\mathcal{I}'}(x) > 0.5 \leftrightarrow D^{\mathcal{I}'}(x) \leq 0.5$. Using the induction hypothesis then gives us $C^{\mathcal{I}'}(x) > 0.5 \leftrightarrow x \notin D^\mathcal{I}$ and as such $C^{\mathcal{I}'}(x) > 0.5 \leftrightarrow x \in C^\mathcal{I}$.
			\item Case $C = \exists R . D$: We have $C^{\mathcal{I}'}(x) > 0.5 \leftrightarrow \exists y \in \Delta^{\mathcal{I}'}: R^{\mathcal{I}'}(x,y)>0.5, D^{\mathcal{I}'}(y)>0.5$. Using the definition of $\mathcal{I}$ and the induction hypothesis gives us $C^{\mathcal{I}'}(x) > 0.5 \leftrightarrow \exists y \in \Delta^\mathcal{I}: (x,y) \in R^\mathcal{I}, y \in D^\mathcal{I}$ and as such $C^{\mathcal{I}'}(x) > 0.5 \leftrightarrow x \in C^\mathcal{I}$.
		\end{itemize}
	\end{proof}
\end{document}